\documentclass[11pt]{article}

\usepackage{mathtools}
\newcommand{\RR}{\mathbb{R}}
\newcommand{\norms}[1]{{\lVert#1\rVert}^2}
\newcommand{\snorms}[1]{{\lVert#1\rVert}_2^2}
\newcommand{\rowmatrix}[4]{
\begin{pmatrix}
#1 & #2 \\
#3 & #4
\end{pmatrix}
}
\newcommand{\keywords}[1]{\emph{Keywords:} #1}

\newcommand{\rank}{\operatorname{rank}}
\newcommand{\card}[1]{\lvert{#1}\rvert}

\newcommand{\NN}{\mathbb{N}}

\usepackage{fullpage,amssymb,amsmath,amsthm,framed,color}
\usepackage{hyperref}

\newtheorem{theorem}{Theorem}
\newtheorem{prop}[theorem]{Proposition}
\newtheorem{lemma}[theorem]{Lemma}
\newtheorem{definition}[theorem]{Definition}

\newtheorem{alg}{Algorithm}

\def\reals{\mathbb{R}}
\def\eps{\epsilon}
\def\suchthat{\;:\;}
\def\given{\;|\;}

\newcommand{\tr}[1]{\operatorname{trace}\left(#1\right)}
\newcommand{\deter}[1]{\operatorname{det}\left(#1\right)}
\newcommand{\linspan}[1]{\operatorname{span}\left(#1\right)}
\newcommand{\vol}[1]{\operatorname{vol}\left(#1\right)}
\newcommand{\conv}[1]{\operatorname{conv}\left(#1\right)}
\newcommand{\norm}[1]{\left\|#1\right\|}
\newcommand{\frob}[1]{\left\|#1\right\|_{F}}
\newcommand{\prob}[1]{{\sf Pr} \left(#1\right)}
\newcommand{\expec}[1]{{\sf E} \left[#1\right]}
\newcommand{\abs}[1]{\left|#1\right|}
\newcommand{\size}[1]{\left|#1\right|}
\newcommand{\inner}[2]{\left\langle#1, #2\right\rangle}

\newcommand{\comment}[1]{{\sf\textcolor{blue}{#1}}}

\title{Efficient volume sampling for row/column subset selection}
\author{Amit Deshpande \\
Microsoft Research India \\
\url{amitdesh@microsoft.com}
\and
Luis Rademacher \\
Computer Science and Engineering \\
Ohio State University \\
\url{lrademac@cse.ohio-state.edu}}

\date{}

\begin{document}
\maketitle

\begin{abstract}
We give efficient algorithms for volume sampling, i.e., for picking $k$-subsets of the rows of any given matrix with probabilities proportional to the squared volumes of the simplices defined by them and the origin (or the squared volumes of the parallelepipeds defined by these subsets of rows).
This solves an open problem from the monograph on spectral algorithms by Kannan and Vempala (see Section $7.4$ of \cite{KV}, also implicit in \cite{BDM, DRVW}).

Our first algorithm for volume sampling $k$-subsets of rows from an $m$-by-$n$ matrix runs in $O(kmn^\omega \log n)$ arithmetic operations and a second variant of it for $(1+\eps)$-approximate volume sampling runs in $O(mn \log m \cdot k^{2}/\eps^{2} + m \log^{\omega} m \cdot k^{2\omega+1}/\eps^{2\omega} \cdot \log(k \eps^{-1} \log m))$ arithmetic operations, which is almost linear in the size of the input (i.e., the number of entries) for small $k$.

Our efficient volume sampling algorithms imply the following results for low-rank matrix approximation:
\begin{enumerate}
\item Given $A \in \reals^{m \times n}$, in $O(kmn^{\omega} \log n)$ arithmetic operations we can find $k$ of its rows such that projecting onto their span gives a $\sqrt{k+1}$-approximation to the matrix of rank $k$ closest to $A$ under the Frobenius norm. This improves the $O(k \sqrt{\log k})$-approximation of Boutsidis, Drineas and Mahoney \cite{BDM} and matches the lower bound shown in \cite{DRVW}. The method of conditional expectations gives a \emph{deterministic} algorithm with the same complexity. The running time can be improved to $O(mn \log m \cdot k^{2}/\eps^{2} + m \log^{\omega} m \cdot k^{2\omega+1}/\eps^{2\omega}  \cdot \log(k \eps^{-1} \log m))$ at the cost of losing an extra $(1+\eps)$ in the approximation factor.
\item The same rows and projection as in the previous point give a $\sqrt{(k+1)(n-k)}$-approximation to the matrix of rank $k$ closest to $A$ under the spectral norm. In this paper, we show an almost matching lower bound of $\sqrt{n}$, even for $k=1$.
\end{enumerate}
\end{abstract}


\section{Introduction}

Volume sampling, i.e., picking $k$-subsets of the rows of any given matrix with probabilities proportional to the squared volumes of the simplicies defined by them, was introduced in \cite{DRVW} in the context of low-rank approximation of matrices. It is equivalent to sampling $k$-subsets of $\{1, \dotsc, m\}$ with probabilities proportional to the corresponding $k$ by $k$ principal minors of any given $m$ by $m$ positive semidefinite matrix.

In the context of low-rank approximation, volume sampling is related to a problem called \emph{row/column-subset selection} \cite{BDM}. Most large data sets that arise in search, microarray experiments, computer vision, data mining etc. can be thought of as matrices where rows and columns are indexed by objects and features, respectively (or vice versa), and we need to pick a small subset of features that are dominant. For example, while studying gene expression data biologists want a small subset of genes that are responsible for a particular disease. Usual dimension reduction techniques such as principal component analysis (PCA) or random projection fail to do this as they typically output singular vectors or random vectors which are linear combinations of a large number of feature vectors. A recent article by Mahoney and Drineas \cite{MD} highlights the limitations of PCA and gives experimental data on practical applications of low-rank approximation based on row/column-subset selection.

\section{Row/column-subset selection and volume sampling}
While dealing with large matrices in practice, we seek smaller or low-dimensional representations of them which are close to them but can be computed and stored efficiently. A popular notion for low-dimensional representation of matrices is low-rank matrices, and the most popular metrics used to measure the closeness of two matrices are the Frobenius or Hilbert-Schmidt norm (i.e., the square root of the sum of squares of entries of their difference) and the spectral norm (i.e., the largest singular value of their difference). The singular value decomposition (SVD) tells us that any matrix $A \in \reals^{m \times n}$ can be written as
\[
A = \sum_{i=1}^{m} \sigma_{i} u_{i} v_{i}^{T},
\]
where $\sigma_{1} \geq \dotsc \geq \sigma_{m} \geq 0$, $u_{i} \in \reals^{m}$ are orthonormal and $v_{i} \in \reals^{n}$ are orthonormal. Moreover, the nearest rank-$k$ matrix to $A$, let us call it $A_{k}$, under both the Frobenius and the spectral norm, is given by
\[
A_{k} = \sum_{i=1}^{k} \sigma_{i} u_{i} v_{i}^{T}.
\]
In other words, the rows of $A_{k}$ are projections of the rows of $A$ onto $\linspan{v_{i} \suchthat 1 \leq i \leq k}$. Because of this, most dimension reduction techniques based on the singular value decomposition, e.g., principal component analysis (PCA), are interpreted as giving $v_{i}$'s as the \emph{dominant} vectors, which happen to be linear combinations of a large number of the rows or feature vectors of $A$.

The \emph{row-subset selection} problem we consider in this paper is: Can we pick a $k$-subset of the rows (or feature vectors) of $A \in \reals^{m \times n}$ so that projecting onto their span is almost as good as projecting onto $\linspan{v_{i} \suchthat 1 \leq i \leq k}$?

Several row-sampling techniques have been considered in the past as an approximate but faster alternative to the singular value decomposition, in the context of streaming algorithms and large data sets that cannot be stored in random access memory \cite{FKV,DMM,DRVW}. The first among these is the squared-length sampling of rows introduced by Frieze, Kannan and Vempala \cite{FKV}. Another sampling scheme due to Drineas, Mahoney and Muthukrishnan \cite{DMM} uses the singular values and singular vectors to decide the sampling probabilities. Later Deshpande, Rademacher, Vempala and Wang \cite{DRVW} introduced volume sampling as a generalization of squared-length sampling.
\begin{definition} \label{def:vol-sampling}
Given $A \in \reals^{m \times n}$, volume sampling is defined as picking a $k$-subset $S$ of $[m]$ with probability proportional to
\[
\deter{A_{S} A_{S}^{T}} = \left(k! \cdot \vol{\conv{\{\bar{0}\} \cup \{a_{i} \suchthat i \in S\}}}\right)^{2},
\]
where $a_{i}$ denotes the $i$-th row of $A$, $A_{S} \in \reals^{k \times n}$ denotes the row-submatrix of $A$ given by rows with indices $i \in S$, and $\conv{\cdot}$ denotes the convex hull.
\end{definition}

The application of volume sampling to low-rank approximation and, more importantly, to the \emph{row-subset selection} problem, is given by the following theorem shown in \cite{DRVW}. It says that picking a subset of $k$ rows according to volume sampling and projecting all the rows of $A$ onto their span gives a $(k+1)$-approximation to the nearest rank-$k$ matrix to $A$.
\begin{theorem} \cite{DRVW} \label{thm:vol-sampling}
Given any $A \in \reals^{m \times n}$,
\[
\expec{\frob{A - \pi_{S}(A)}^{2}} \leq (k+1) \frob{A - A_{k}}^{2},
\]
when $S$ is picked according to volume sampling, $\pi_{S}(A) \in \reals^{m \times n}$ denotes the matrix obtained by projecting all the rows of $A$ onto $\linspan{a_{i} \suchthat i \in S}$, and $A_{k}$ is the matrix of rank $k$ closest to $A$ under the Frobenius norm.
\end{theorem}
As we will see later, this easily implies
\[
\expec{\norm{A - \pi_{S}(A)}_{2}} \leq \sqrt{(k+1)(n-k)} \norm{A - A_{k}}_{2}.
\]
Theorem \ref{thm:vol-sampling} gives only an existence result for row-subset selection and we also know a matching lower bound that says this is the best we can possibly do.
\begin{theorem} \cite{DRVW} \label{thm:DRVW-lower}
For any $\eps > 0$, there exists a matrix $A \in \reals^{(k+1) \times k}$ such that picking any $k$-subset $S$ of its rows gives
\[
\frob{A - \pi_{S}(A)} \geq (1-\eps) \sqrt{k+1} \frob{A - A_{k}}.
\]
\end{theorem}

However, no efficient algorithm was known for volume sampling prior to this work. An algorithm mentioned in Deshpande and Vempala \cite{DV} does $k!$-approximate volume sampling in time $O(kmn)$, which means that plugging it in Theorem \ref{thm:vol-sampling} can only guarantee $(k+1)!$-approximation instead of $(k+1)$. Finding an efficient algorithm for volume sampling is mentioned as an open problem in the recent monograph on spectral algorithms by Kannan and Vempala (see Section $7.4$ of \cite{KV}).

Boutsidis, Drineas and Mahoney \cite{BDM} gave an alternative approach to row-subset selection (without going through volume sampling) and here is a re-statement of the main theorem from their paper which uses columns instead of rows.


\begin{theorem} \cite{BDM} \label{thm:bdm-soda}
For any $A \in \reals^{m \times n}$, a $k$-subset $S$ of its rows can be found in time $O\left(\min\{mn^{2}, m^{2}n\}\right)$ such that
\begin{align*}
\frob{A - \pi_{S}(A)} & = O(k \sqrt{\log k}) \frob{A - A_{k}} \\
\norm{A - \pi_{S}(A)}_{2} & = O\left(k^{3/4} (n-k)^{1/4} \sqrt{\log k}\right) \norm{A - A_{k}}_{2}.
\end{align*}
\end{theorem}

Row/column-subset selection problem is related to rank-revealing decompositions considered in linear algebra \cite{GE,P}, and the previous best algorithmic result for row-subset selection in the spectral norm case was given by a result of Gu and Eisenstat \cite{GE} on strong rank-revealing QR decompositions. The following theorem is a direct consequence of \cite{GE} as pointed out in \cite{BDM}.
\begin{theorem}
Given $A \in \reals^{m \times n}$, an integer $k \leq n$ and $f \geq 1$, there exists a $k$-subset $S$ of the columns of $A$ such that
\[
\norm{A^{T} - \pi_{S}(A^{T})}_{2} \leq \sqrt{1 + f^{2} k(n-k)} \norm{A - A_{k}}_{2}.
\]
Moreover, this subset $S$ can be found in time $O\left((m + n \log_{f}n)n^{2}\right)$.
\end{theorem}
In the context of volume sampling, it is interesting to note that Pan \cite{P} has used an idea of picking submatrices of \emph{locally maximum volume} (or determinants) for rank-revealing matrix decompositions. We refer the reader to \cite{P} for details.

The results of Goreinov, Tyrtyshnikov and Zamarashkin \cite{GT,GTZ} on pseudo-skeleton approximations of matrices look at submatrices of maximum determinants as good candidates for row/column-subset selection.
\begin{theorem} \cite{GT}
If $A \in \reals^{m \times n}$ can be written as
\[
\left(\begin{array}{cc} A_{11} & A_{12} \\ A_{21} & A_{22} \end{array}\right),
\]
where $A_{11} \in \reals^{k \times k}$ is the $k$ by $k$ submatrix of $A$ of maximum determinant. Then,
\[
\max_{i,j} \abs{(A_{22} - A_{12}A_{11}^{-1}A_{21})_{ij}} \leq (k+1) \norm{A - A_{k}}_{2}.
\]
\end{theorem}

Because of this relation between row/column-subset selection and the related ideas about picking submatrices of maximum volume, \c{C}ivril and Magdon-Ismail \cite{CM1,CM2} looked at the problem of picking a $k$-subset $S$ of rows of a given matrix $A \in \reals^{m \times n}$ such that $\deter{A_{S} A_{S}^{T}}$ is maximized. They show that this problem is NP-hard \cite{CM1} and moreover, it is NP-hard to even approximate it within a factor of $2^{ck}$, for some constant $c>0$ \cite{CM2}. This is interesting in the light of our results because we show that even though finding the row-submatrix of maximum volume is NP-hard, we can still sample them with probabilities proportional to their volumes in polynomial time.


\subsection{Our results}
Our main result is a polynomial time algorithm for exact volume sampling. In Section \ref{sec:algo-outline}, we give an outline of our Algorithm \ref{alg:outline}, followed by two possible subroutines given by Algorithms  \ref{alg:omegaSub} and \ref{alg:svdSub} that could be plugged into it.

\begin{theorem}[polynomial-time volume sampling]\label{thm:polytimeVS}
The randomized algorithm given by the combination of the algorithm outlined in Algorithm \ref{alg:outline} with Algorithm \ref{alg:omegaSub} as its subroutine, when given a matrix $A \in \reals^{m \times n}$ and an integer $1 \leq k \leq \rank(A)$, outputs a random $k$-subset of the rows of $A$ according to volume sampling, using $O(kmn^\omega \log n)$ arithmetic operations.
\end{theorem}
The basic idea of the algorithm is as follows: instead of picking a $k$-subset, pick an ordered $k$-tuple of rows according to volume sampling (i.e., volume sampling suitably extended to all $k$-tuples such that for any fixed $k$-subset, all its $k!$ permutations are all equally likely). We observe that the marginal distribution of the first coordinate of such a random tuple can be expressed in terms of coefficients of the characteristic polynomials of $AA^{T}$ and $B_{i} B_{i}^{T}$, where $B_{i} \in \reals^{m \times n}$ is the matrix obtained by projecting each row of $A$ orthogonal to the $i$-th row $a_{i}$. Using this interpretation, it is easy to sample the first index of the $k$-tuple with the right marginal probability. Now we project the rows of $A$ orthogonal to the chosen row and repeat to pick the next row, until we have picked $k$ of them.

The algorithm just described informally, if implemented as stated, would have a polynomial dependence in $m$, $n$ and $k$, for some low-degree polynomial. We can do better and get a linear dependence in $m$ by working with $A^{T}A$ in place of $AA^{T}$ and computing the projected matrices using rank-$1$ updates (Theorem \ref{thm:polytimeVS}), while still having a polynomial time guarantee and sampling exactly. It would be even faster to perform rank-1 updates to the characteristic polynomial itself, but that requires the computation of the inverse of a polynomial matrix (Proposition \ref{prop:svdVS}), and it is not clear to us at this time that there is a fast enough exact algorithm that works for arbitrary matrices. Jeannerod and Villard \cite{JV} give an algorithm to invert a \emph{generic} $n$-by-$n$ matrix with entries of degree $d$, with $n$ a power of two, in time $O(n^3 d)$. This would lead to the computation of all marginal probabilities for one row in time $O(n^3 + mn^2)$ (a variation of Algorithm \ref{alg:svdSub} and its analysis). 

Instead, if we are willing to be more practical, while sacrificing our guarantees, then we can perform rank-$1$ updates to the characteristic polynomial by using the singular value decomposition (SVD). In \cite{GVL}, an algorithm with cost $O(\min\{mn^{2}, m^{2}n\})$ arithmetic operations is given for the singular value decomposition but the SVD cannot be computed \emph{exactly} and we do not know how its error propagates in our algorithm which uses many such computations. If the SVD of an $m$-by-$n$ matrix can be computed in time $O(T_{svd})$, this leads to a nearly-exact algorithm for volume sampling in time $O(k T_{svd} + kmn^{2})$. See Proposition \ref{prop:svdVS} for details.

Volume sampling was originally defined in \cite{DRVW} to prove Theorem \ref{thm:vol-sampling}, in particular, to show that any matrix $A$ contains $k$ rows in whose span lie the rows of a rank-$k$ approximation to $A$ that is no worse than the best in the Frobenius norm. Efficient volume sampling leads to an efficient selection of $k$ rows that satisfy this guarantee, \emph{in expectation}. In Section \ref{sec:derand}, we use the method of conditional expectations to derandomize this selection. This gives an efficient deterministic algorithm (Algorithm \ref{alg:derand}) for row-subset selection with the following guarantee in the Frobenius norm. This guarantee immediately implies a guarantee in the spectral norm, as follows:

\begin{theorem}[deterministic row subset selection] \label{thm:det-sub-select}
Deterministic Algorithm \ref{alg:derand}, when given a matrix $A \in \reals^{m \times n}$ and an integer $1 \leq k \leq \rank(A)$, outputs a $k$-subset $S$ of the rows of $A$, using $O(kmn^{\omega} \log n)$ arithmetic operations, such that
\begin{align*}
\frob{A - \pi_{S}(A)} & \leq \sqrt{k+1} \frob{A - A_{k}} \\
\norm{A - \pi_{S}(A)}_{2} & \leq \sqrt{(k+1)(n-k)} \norm{A - A_{k}}_{2}.
\end{align*}
\end{theorem}
This improves the $O(k \sqrt{\log k})$-approximation of Boutsidis, Drineas and Mahoney \cite{BDM} for the Frobenius norm case and matches the lower bound shown in Theorem \ref{thm:DRVW-lower} due to \cite{DRVW}.

The superlinear dependence on $n$ might be too slow for some applications, while it might be acceptable to perform volume sampling or row/column-subset selection approximately. Our volume sampling algorithm (Algorithm \ref{alg:outline}) can be made faster, while losing on the exactness, by using the idea of random projection that preserves volumes of subsets. Magen and Zouzias \cite{MZ} have the following generalization of the Johnson-Lindenstrauss lemma: for $m$ points in $\reals^{n}$ there exists a random projection of them into $\reals^{d}$, where $d = O\left(k^{2} \log m / \eps^{2}\right)$, that preserves the volumes of simplices formed by subsets of $k$ or fewer points within $1\pm\eps$. Therefore, we get a $(1\pm\epsilon)$-approximate volume sampling algorithm that requires $O(mnd)$-time to do the random projection (by matrix multiplication) and then $O(mnd^{\omega} \log d)$ time for volume sampling on the new $m$-by-$d$ matrix (according to Theorem \ref{thm:polytimeVS}).

\begin{theorem}[fast volume sampling] \label{thm:fast-VS}
Using random projection for dimensionality reduction, the polynomial time algorithm for volume sampling mentioned in Theorem \ref{thm:polytimeVS} (i.e., Algorithm \ref{alg:outline} with Algorithm \ref{alg:omegaSub} as its subroutine), gives $(1+\eps)$-approximate volume sampling, using
\[
O\left(mn \log m \cdot \frac{k^{2}}{\eps^{2}} + m \log^{\omega} m \cdot \frac{k^{2\omega + 1}}{\eps^{2\omega}} \log(k \eps^{-1} \log m)\right).
\]
arithmetic operations.
\end{theorem}

Finally, we show a lower bound for row/column-subset selection in the spectral norm that almost matches our upper bound in terms of the dependence on $n$.
\begin{theorem}[lower bound] \label{thm:lower-bound}
There exists a matrix $A \in \reals^{n \times (n+1)}$ such that
\[
\norm{A - \pi_{\{i\}}(A)}_{2} = \Omega(\sqrt{n}) \norm{A - A_{1}}_{2}, \quad \text{for all $1 \leq i \leq n$},
\]
where $\pi_{\{i\}}(A) \in \reals^{n \times (n+1)}$ is the matrix obtained by projecting each row of $A$ onto the span of its $i$-th row $a_{i}$.
\end{theorem}


\section{Preliminaries and notation}\label{sec:prelim}
For $m \in \NN$, let $[m]$ denote the set $\{1, \dotsc, m\}$. For any matrix $A \in \reals^{m \times n}$, we denote its rows by $a_{1}, a_{2}, \dots, a_{m} \in \reals^{n}$. For $S \subseteq [m]$, let $A_{S}$ be the row-submatrix of $A$ given by the rows with indices in $S$. By $\linspan{S}$ we denote the linear span of $\{a_{i} \suchthat i \in S\}$ and let $\pi_{S}(A) \in \reals^{m \times n}$ be the matrix obtained by projecting each row of $A$ onto $\linspan{S}$. Hence, $A - \pi_{S}(A) \in \reals^{m \times n}$ is the matrix obtained by projecting each row of $A$ orthogonal to $\linspan{S}$.

Throughout the paper we assume $m \geq n$. This assumption is not needed most of the time, but justifies sometimes working with $A^T A$ instead of $A A^T$ and, more generally, some choices in the design of our algorithms. It is also partially justified by our use of a random projection as a preprocessing step that makes $n$ small.

The singular values of $A \in \reals^{m \times n}$ are defined as the positive square-roots of the eigenvalues of $AA^{T} \in \reals^{m \times m}$ (or $A^{T}A \in \reals^{n \times n}$, up to some extra singular values equal to zero), and we denote them by $\sigma_{1} \geq \sigma_{2} \geq \dotsb \geq \sigma_{m} \geq 0$. Well-known identities of the singular values like
\[
\tr{AA^{T}} = \sum_{i=1}^{m} \sigma_{i}^{2} \quad \text{and} \quad \deter{AA^{T}} = \prod_{i=1}^{m} \sigma_{i}^{2}
\]
can be generalized into the following lemma.


\begin{lemma} (Proposition 3.2 in \cite{DRVW}) \label{lemma:char-poly}
For any $A \in \reals^{m \times n}$,
\[
\sum_{S \subseteq [m] \suchthat \size{S}=k} \deter{A_{S} A_{S}^{T}} = \sum_{i_{1} < \dotsb < i_{k}} \sigma_{i_{1}}^{2} \dotsb \sigma_{i_{k}}^{2} = \abs{c_{m-k}(AA^{T})},
\]
where $\sigma_{1}, \dotsc, \sigma_{m}$ are the singular values of $A$, i.e., eigenvalues of $AA^{T}$, and
\[
\deter{xI - AA^{T}} = x^{m} + c_{m-1}(AA^{T}) x^{m-1} + \dotsc + c_{0}(AA^{T}) = \prod_{i=1}^{m} (x - \sigma_{i}^{2}),
\]
is the characteristic polynomial of $AA^{T}$. Using $\deter{xI - AA^{T}} = x^{m-n} \deter{xI - A^{T}A}$, we can alternatively use $c_{m-k} (AA^{T}) = c_{n-k}(A^{T}A)$ in the above formula, for $k \leq n$.
\end{lemma}

Let $\omega$ be the exponent of the arithmetic complexity of matrix multiplication. We use that there is an algorithm for computing the characteristic polynomial of an $n$-by-$n$ matrix using $O(n^\omega \log n)$ arithmetic operations \cite[Section 16.6]{ACT}.

Here is another lemma that we will need about dividing determinants into products of two determinants.
\begin{lemma} \label{lemma:det-division}
Let $A \in \reals^{m \times n}$, $S, T \subseteq [m]$, $S \cap T = \emptyset$ and $B = A - \pi_{S}(A)$. Then
\[
\deter{A_{S \cup T} A_{S \cup T}^{T}} = \deter{A_{S} A_{S}^{T}} \deter{B_{T} B_{T}^{T}}.
\]
\end{lemma}
\begin{proof}
Without loss of generality, we can reduce ourselves to the case where $S \cup T$ is all rows of the given matrix: Let $C= A_{S \cup T} \in \RR^{\card{S \cup T} \times n}$,  $D = C -\pi_{S}(C)$. We have $D = B_{S \cup T}$. Then what we want to prove can be rewritten as:
\[
\det (C C^T) = \det(C_{S} C_{S}^T) \det(D_T D_T^T).
\]
To show this, we consider two cases. If $C_S C_S^T$ is singular, then both sides of the equality are zero. If $C_S C_S^T$ is invertible, then we can perform block Gaussian elimination and write
\[
\rowmatrix{E}{F}{G}{H} \rowmatrix{I}{-E^{-1} F}{0}{I}
= \rowmatrix{E}{0}{G}{D-G E^{-1}F},
\]
applied to $\rowmatrix{E}{F}{G}{H}= C$. Writing the determinants of the block-triangular matrices gives
\[
\det (C C^T) = \det (C_S C_S^T) \det (C_T C_T^T - C_T C_S^T(C_S C_S^T)^{-1} C_S C_T^T).
\]
Now, the projection of the rows of a matrix $K$ onto the row-space of a matrix $L$ can be written as
\[
\pi_{L}(K) = K L^T (L L^T)^{-1} L,
\]
so that $D_T = C_T - C_T C_S^T (C_S C_S^T)^{-1} C_S$,  and
\[
D_T D_T^T = C_T C_T^T - C_T C_S^T(C_S C_S^T)^{-1} C_S C_T^T.
\]
This completes the proof.
\end{proof}

Finally, a well-known lemma about how the determinant of a matrix changes under a rank-$1$ update.
\begin{lemma}[matrix determinant lemma] \label{lemma:matrix-det}
For any invertible $M \in \reals^{m \times m}$ and $u, v \in \reals^{m}$,
\[
\deter{M + uv^{T}} = (1 + v^{T}M^{-1}u) \deter{M}.
\]
\end{lemma}

\section{Efficient volume sampling algorithms} \label{sec:algo-outline}
We first outline our volume sampling algorithm to convince the reader that volume sampling can be done in polynomial time. In the subsequent subsections, we give improved subroutines to get faster implementations of the same idea.

The main idea behind our algorithm is based on Lemma \ref{lemma:marginals} about the marginal probabilities encountered in volume sampling. To explain this, it is more convenient to look at volume sampling defined as a distribution on $k$-tuples $(X_{1}, X_{2}, \dotsc, X_{k})$ instead of $k$-subsets, where each of the $k!$ permutations of a $k$-subset is equally likely, i.e., for any $(i_{1}, i_{2}, \dotsc, i_{k}) \in [m]^{k}$,
\[
\prob{X_{1} = i_{1}, \dotsc, X_{k} = i_{k}} = \begin{cases} \dfrac{\deter{A_{\{i_{1}, \dotsc, i_{k}\}} A_{\{i_{1}, \dotsc, i_{k}\}}^{T}}}{k! \sum_{S \subseteq [m] \suchthat \size{S}=k} \deter{A_{S} A_{S}^{T}}} & \text{if $i_{1}, \dotsc, i_{k}$ are distinct} \\ \quad & \quad \\ 0 & \text{otherwise} \end{cases}
\]
Then the marginal probabilities $\prob{X_{t} = i \given X_{1} = i_{1}, \dotsc, X_{t-1} = i_{t-1}}$ have the following interpretation in terms of the coefficients of certain characteristic polynomials.
\begin{lemma} \label{lemma:marginals}
Let $(i_{1}, \dotsc, i_{t-1}) \in [m]^{t-1}$ such that $\prob{X_{1} = i_{1}, \dotsc, X_{t-1} = i_{t-1}} > 0$, for a random $k$-tuple $(X_{1}, X_{2}, \dotsc, X_{k})$ from the extended volume sampling over $k$-tuples. Let $S = \{i_{1}, \dotsc, i_{t-1}\}$, $B = A - \pi_{S}(A)$ and $C_{i} = B - \pi_{\{i\}}(B) = A - \pi_{S \cup \{i\}}(A)$. Then, \[
\prob{X_{t} = i \given X_{1} = i_{1}, \dotsc, X_{t-1} = i_{t-1}} = \frac{\norm{b_{i}}^{2} \abs{c_{m-k+t}(C_{i} C_{i}^{T})}}{(k-t+1) \abs{c_{m-k+t-1}(BB^{T})}}.
\]
\end{lemma}
\begin{proof}
\begin{align*}
& \prob{X_{t} = i \given X_{1} = i_{1}, \dotsc, X_{t-1} = i_{t-1}} \\
& = \frac{\sum_{(i_{t+1}, \dotsc, i_{k}) \in [m]^{k-t}} \prob{X_{1} = i_{1}, \dotsc, X_{t-1} = i_{t-1}, X_{t} = i, X_{t+1} = i_{t+1}, \dotsc, X_{k} = i_{k}}}{\sum_{(i_{t}, \dotsc, i_{k}) \in [m]^{k-t+1}} \prob{X_{1} = i_{1}, \dotsc, X_{t-1} = i_{t-1}, X_{t} = i_{t}, X_{t+1} = i_{t+1}, \dotsc, X_{k} = i_{k}}} \\
& = \frac{(k-t)! \sum_{T \subseteq [m] \suchthat \size{S \cup \{i\} \cup T}=k, \size{T}=k-t} \deter{A_{S \cup \{i\} \cup T} A_{S \cup \{i\} \cup T}^{T}}}{(k-t+1)! \sum_{T \subseteq [m] \suchthat \size{S \cup T}=k, \size{T}=k-t+1} \deter{A_{S \cup T} A_{S \cup T}^{T}}} \\
& = \frac{\sum_{T \subseteq [m] \suchthat \size{S \cup\{i\} \cup T}=k, \size{T}=k-t} \deter{A_{S} A_{S}^{T}} \deter{B_{\{i\} \cup T} B_{\{i\} \cup T}^{T}}}{(k-t+1) \sum_{T \subseteq [m] \suchthat \size{S \cup T}=k, \size{T}=k-t+1} \deter{A_{S} A_{S}^{T}} \deter{B_{T} B_{T}^{T}}} \quad \text{by Lemma \ref{lemma:det-division}} \\
& = \frac{\sum_{T \subseteq [m] \suchthat \size{S \cup \{i\} \cup T}=k, \size{T}=k-t} \norm{b_{i}}^{2} \deter{C_{T} C_{T}^{T}}}{(k-t+1) \sum_{T \subseteq [m] \suchthat \size{S \cup T}=k, \size{T}=k-t+1} \deter{B_{T} B_{T}^{T}}} \quad \text{by Lemma \ref{lemma:det-division} applied to $B$} \\
& = \frac{\norm{b_{i}}^{2} \sum_{T \subseteq [m] \suchthat \size{T}=k-t} \deter{C_{T} C_{T}^{T}}}{(k-t+1) \sum_{T \subseteq [m] \suchthat \size{T}=k-t+1} \deter{B_{T} B_{T}^{T}}} \quad \text{since the extra terms in the sum are all zero} \\
& = \frac{\norm{b_{i}}^{2} \abs{c_{m-k+t}(C_{i} C_{i}^{T})}}{(k-t+1) \abs{c_{m-k+t-1}(BB^{T})}} \quad \text{by Lemma \ref{lemma:char-poly}}.
\end{align*}
\end{proof}

With this lemma in hand, let us consider the following outline of our algorithm. We will later give two more efficient implementations of this outline, depending on how the $p_{i}$'s are computed.

\begin{framed}
\begin{alg}\label{alg:outline}
{\bf Outline of our volume sampling algorithm}
\end{alg}
\noindent Input: a matrix $A \in \reals^{m \times n}$ and $1 \leq k \leq \rank(A)$.

\noindent Output: a subset $S$ of $k$ rows of $A$ picked with probability proportional to $\deter{A_{S}A_{S}^{T}}$.

\begin{enumerate}
\item Initialize $S \leftarrow \emptyset$ and $B \leftarrow A$. For $t=1$ to $k$ do:
\begin{enumerate}
\item For $i=1$ to $m$ compute:
\[
p_{i} = \norm{b_{i}}^{2} \cdot \abs{c_{m-k+t} (C_{i} C_{i}^{T})},
\]
where $C_{i} = B - \pi_{\{i\}}(B)$ is a matrix obtained by projecting each row of $B$ orthogonal to $b_{i}$.
\item Pick $i$ with probability proportional to $p_{i}$. Let $S \leftarrow S \cup \{i\}$ and $B \leftarrow C_{i}$.
\end{enumerate}
\item Output $S$.
\end{enumerate}
\end{framed}

Now we show the correctness of the algorithm:
\begin{prop} \label{prop:outline-correct}
The probability that our volume sampling algorithm outlined above picks a $k$-subset $S$ is proportional to $\deter{A_{S} A_{S}^{T}}$. This algorithm can be implemented with a cost of $O(k m^3 n + km^{\omega+1} \log m)$ arithmetic operations.
\end{prop}
\begin{proof}
By Lemma \ref{lemma:marginals}, for any $i_{1}, i_{2}, \dotsc, i_{k}$ such that $\prob{X_{1} = i_{1}, \dotsc, X_{k} = i_{k}}$, the probability that our algorithm picks a sequence of rows indexed $i_{1}, i_{2}, \dotsc, i_{k}$ in that order is equal to
\[
\prod_{t=1}^{k} \prob{X_{t} = i_{t} \given X_{1} = i_{1}, \dotsc, X_{t-1} = i_{t-1}} = \prob{X_{1} = i_{1}, \dotsc, X_{k} = i_{k}} = \frac{\deter{A_{\{i_{1}, \dotsc, i_{k}\}} A_{\{i_{1}, \dotsc, i_{k}\}}^{T}}}{k! \sum_{S \subseteq [m] \suchthat \size{S}=k} \deter{A_{S} A_{S}^{T}}}.
\]
Otherwise, the probability is zero because in
the execution of the algorithm,
$\norm{b_{i}} = 0$ for some step $t$. This proves the correctness of our algorithm.

Given that one can compute the characteristic polynomial of an $m$-by-$m$ matrix in $O(m^\omega \log m)$ (see Section \ref{sec:prelim}), our outline can be implemented with the following count of arithmetic operations: for every $t$ and $i$, $O(m^2 n)$ to compute $C_i C_i^T$, $O(m^2 n + m^\omega \log m)$ in total for $p_i$. Thus, volume sampling in $O(k m^3 n + km^{\omega +1} \log m)$.
\end{proof}

\subsection{Efficient volume sampling without SVD} \label{subsec:no-svd}
Here we present the first (faster) subroutine for computing the marginal probabilities $p_{i}$'s within the volume sampling algorithm outlined in Section \ref{sec:algo-outline}. The two main ideas behind this subroutine are: (1) We can work with $B^{T}B, C_{i}^{T} C_i \in \reals^{n \times n}$ instead of $BB^{T}, C_{i} C_{i}^{T} \in \reals^{m \times m}$. Assuming $m \geq n$, this saves on running time. (2) Each $C_{i}$ is a rank-$1$ update of $B$ and therefore, once we have $B^{T}B$, it can be used to compute all $C_{i}^{T} C_{i}$ efficiently.

\begin{framed}
\begin{alg}\label{alg:omegaSub}
{\bf First subroutine for marginal probabilities}
\end{alg}
\noindent Input: $B \in \reals^{m \times n}$.

\noindent Output: $p_{1}, p_{2}, \dotsc, p_{m}$. \\

For $i=1$ to $m$ do:
\begin{enumerate}
\item Compute the matrix $C_{i}^T C_{i} \in \reals^{n \times n}$ by the following formula
\[
C_{i}^{T} C_{i} = B^{T} B - \frac{B^{T} B b_{i} b_{i}^{T}}{\norm{b_{i}}^{2}} - \frac{b_{i} b_{i}^{T} B^{T} B}{\norm{b_{i}}^{2}} + \frac{b_{i} b_{i}^{T} B^{T} B b_{i} b_{i}^{T}}{\norm{b_{i}}^{4}}.
\]
\item Compute the characteristic polynomial of $C_{i}^{T} C_{i}$ and output
\[
p_{i} = \norm{b_{i}}^{2} \cdot \abs{c_{n-k+t} (C_{i}^{T} C_{i})}.
\]
\end{enumerate}
\end{framed}

\begin{prop} \label{prop:omegaVS}
For any given $B \in \reals^{m \times n}$, the Algorithm \ref{alg:omegaSub} above computes $p_{1}, \dotsc, p_{m}$ in $O(mn^{\omega} \log n)$ arithmetic operations.
\end{prop}
\begin{proof}
$B^{T}B$ can be computed in time $O(mn^{2})$. Observe that since $C_{i}$ is obtained by projecting each row of $B$ orthogonal to $b_{i}$,
\[
C_{i} = B - \frac{1}{\norm{b_{i}}^{2}} Bb_{i}b_{i}^{T},
\]
and therefore,
\[
C_{i}^{T} C_{i} = B^{T} B - \frac{B^{T} B b_{i} b_{i}^{T}}{\norm{b_{i}}^{2}} - \frac{b_{i} b_{i}^{T} B^{T} B}{\norm{b_{i}}^{2}} + \frac{b_{i} b_{i}^{T} B^{T} B b_{i} b_{i}^{T}}{\norm{b_{i}}^{4}}.
\]
So once we have $B^{T}B$, for each $i$, $C_{i}^{T} C_{i}$ can be computed in time $O(n^{2})$ and the characteristic polynomial of $C_{i}^{T} C_{i}$ can be computed in time $O(n^{\omega} \log n)$ \cite[Section 16.6]{ACT}. By Lemma \ref{lemma:char-poly}, $c_{m-k+t} (C_{i} C_{i}^{T}) = c_{n-k+t} (C_{i}^{T} C_{i})$ and hence, the above subroutine results into an $O(kmn^{\omega} \log n)$ time algorithm for volume sampling.
\end{proof}

\begin{theorem}[same as Theorem \ref{thm:polytimeVS}]
The randomized algorithm given by the combination of the algorithm outlined in Algorithm \ref{alg:outline} with Algorithm \ref{alg:omegaSub} as its subroutine, when given a matrix $A \in \reals^{m \times n}$ and an integer $1 \leq k \leq \rank(A)$, outputs a random $k$-subset of the rows of $A$ according to volume sampling, using $O(kmn^\omega \log n)$ arithmetic operations.
\end{theorem}
\begin{proof}
The proof follows by combining Proposition \ref{prop:outline-correct} and Proposition \ref{prop:omegaVS}, and since we compute all the $p_{i}$'s simultaneously in each round in $O(mn^{\omega} \log n)$ arithmetic operations, the total number of arithmetic operations is $O(kmn^{\omega} \log n)$.
\end{proof}

\subsection{Efficient volume sampling using SVD} \label{subsec:svd}
Taking further the idea that each $C_{i}$ is a rank-$1$ update of $B$, we can give a faster algorithm based on the singular value decomposition of $B$. Given the singular value decomposition of a matrix and using the matrix determinant lemma (Lemma \ref{lemma:matrix-det}), one can give a precise formula for how the characteristic polynomial changes under a rank-$1$ update. Using this subroutine in the volume sampling algorithm outlined in Section \ref{sec:algo-outline} we get an algorithm for nearly-exact volume sampling (depending on the precision of the computed SVD) in time $O(kT_{svd} + kmn^{2})$, where $T_{svd}$ is the running time of SVD on an $m$-by-$n$ matrix.

\begin{framed}
\begin{alg}\label{alg:svdSub}
{\bf Second subroutine for marginal probabilities.}
\end{alg}
\noindent Input: $B \in \reals^{m \times n}$. \\
\noindent Output: $p_{1}, p_{2}, \dotsc, p_{m}$.

\begin{enumerate}
\item Compute the (thin) singular value decomposition $B=U \Sigma V^{T}$, say $U \in \reals^{m \times n}$ and $\Sigma, V \in \reals^{n \times m}$, and keep the singular values $\sigma_{1}, \sigma_{2}, \dotsc, \sigma_{n}$ and define $\sigma_{n+1} = \dotsc = \sigma_{m} = 0$. Also keep the columns of $U$, i.e., the left singular vectors $u_{1}, u_{2}, \dotsc, u_{n} \in \reals^{m}$.
\item Compute the polynomial products
\begin{align*}
f(x) & = \prod_{l=1}^{m} (x - \sigma_{l}^{2}), \quad \text{and} \\
g_{j}(x) & = \prod_{l \neq j} (x - \sigma_{l}^{2}), \quad \text{for all $1 \leq j \leq m$.}
\end{align*}
\item For $i=1$ to $m$ output:
\[
p_{i} = \norm{b_{i}}^{2} \cdot \abs{\text{coefficient of $x^{m-k+t}$ in $f(x) + \dfrac{1}{\norm{b_{i}}^{2}} \sum_{j=1}^{n} \sigma_{j}^{2} (u_{j})_{i}^{2} g_{j}(x)$}}.
\]
\end{enumerate}
\end{framed}
\begin{prop} \label{prop:svdVS}
In the real arithmetic model and given exact $U$ and $\Sigma$, using the Algorithm \ref{alg:svdSub} as a subroutine inside Algorithm \ref{alg:outline} outlined for volume sampling, we get an algorithm for volume sampling. If $T_{svd}$ is the running time for computing the singular value decomposition of $m$-by-$n$ matrices, the algorithm runs in time $O(k T_{svd} + kmn^2)$.
\end{prop}

%
%
\begin{proof}
Using the matrix determinant lemma (Lemma \ref{lemma:matrix-det}), the characteristic polynomial of $C_{i} C_{i}^{T}$ can be written as
\begin{align*}
\deter{xI - C_{i} C_{i}^{T}} & = \deter{xI - BB^{T} + \frac{1}{\norm{b_{i}}^{2}} (Bb_{i})(Bb_{i})^{T}} \\
& = \left(1 + \frac{1}{\norm{b_{i}}^{2}} b_{i}^{T} B^{T} (xI - BB^{T})^{-1} Bb_{i}\right) \deter{xI - BB^{T}} \\
& = \left(1 + \frac{1}{\norm{b_{i}}^{2}} b_{i}^{T} B^{T} (xI - \tilde{U}\tilde{\Sigma}^{2}\tilde{U}^{T})^{-1} Bb_{i}\right) \deter{xI - BB^{T}} \\
& \quad \text{by extending $U, \Sigma, V$ to get $B = \tilde{U} \tilde{\Sigma} \tilde{V}^{T}$ with $\tilde{U}, \tilde{\Sigma} \in \reals^{m \times m}$ and $\tilde{V} \in \reals^{m \times n}$} \\
& = \left(1 + \frac{1}{\norm{b_{i}}^{2}} b_{i}^{T} B^{T}\tilde{U} (xI - \tilde{\Sigma}^{2})^{-1} \tilde{U}^{T}B b_{i}\right) \deter{xI - BB^{T}} \\
& = \left(1 + \frac{1}{\norm{b_{i}}^{2}} b_{i}^{T} \tilde{V}\tilde{\Sigma}^{T} (xI - \tilde{\Sigma}^{2})^{-1} \tilde{\Sigma} \tilde{V}^{T}b_{i}\right) \deter{xI - BB^{T}} \\
& = \left(1 + \frac{1}{\norm{b_{i}}^{2}} \sum_{j=1}^{m} \frac{\sigma_{j}^{2} (\tilde{u}_{j})_{i}^{2}}{x - \sigma_{j}^{2}}\right) \prod_{l=1}^{m} (x - \sigma_{l}^{2}) \\
& = \left(1 + \frac{1}{\norm{b_{i}}^{2}} \sum_{j=1}^{n} \frac{\sigma_{j}^{2} (u_{j})_{i}^{2}}{x - \sigma_{j}^{2}}\right) \prod_{l=1}^{m} (x - \sigma_{l}^{2}) \\
& = \prod_{l=1}^{m} (x - \sigma_{l}^{2}) + \frac{1}{\norm{b_{i}}^{2}} \sum_{j=1}^{n} \sigma_{j}^{2} (u_{j})_{i}^{2} \prod_{l \neq j} (x - \sigma_{l}^{2}) \\
& = f(x) + \frac{1}{\norm{b_{i}}^{2}} \sum_{j=1}^{n} \sigma_{j}^{2} (u_{j})_{i}^{2} g_{j}(x).
\end{align*}
Thus,
\[
c_{m-k+t} (C_{i} C_{i}^{T}) = \text{coefficient of $x^{m-k+t}$ in $f(x) + \dfrac{1}{\norm{b_{i}}^{2}} \sum_{j=1}^{n} \sigma_{j}^{2} (u_{j})_{i}^{2} g_{j}(x)$}.
\]
Once we have the singular value decomposition of $B$, $f(x)$ and $g_{j}(x)$ can all be computed in time $O(n^{2})$ using polynomial products. This is because there are at most $n$ non-zero $\sigma_{i}$'s. Thus, $f(x)$ and all the $g_{j}(x)$ for $1 \leq j \leq m$ can be computed in time $O(mn^{2})$ and then using the above formula we get $c_{m-k+t} (C_{i} C_{i}^{T})$.
\end{proof}

\subsection{Approximate volume sampling in nearly linear time}\label{subsec:random-proj}
Magen and Zouzias \cite{MZ} showed that the random projection lemma of Johnson and Lindenstrauss can be generalized to preserve volumes of subsets after embedding. Here is a restatement of Theorem 1 of \cite{MZ} using $O(\eps/k)$ instead of $\epsilon$ in their original statement.
\begin{theorem} \label{thm:random-projection} \cite{MZ}
For any $A \in \reals^{m \times n}$, $1\leq k \leq n$ and $0 < \eps \leq 1/2$, there is
\[
d = O\left(\frac{k^{2} \log m}{\eps^{2}}\right),
\]
and there is a mapping $f: \reals^{n} \rightarrow \reals^{d}$ such that
\[
\deter{A_{S} A_{S}^{T}} \leq \deter{\tilde{A}_{S} \tilde{A}_{S}^{T}} \leq (1+\eps) \deter{A_{S} A_{S}^{T}},
\]
for all $S \subseteq [m]$ such that $\card{S} \leq k$, where $\tilde{A} \in \reals^{m \times d}$ has its $i$-th row as $f(a_{i})$. Moreover, $f$ is a linear mapping given by multiplication with a random $n$ by $d$ matrix with i.i.d. Gaussian entries, so computing $\tilde{A}$ takes time $O(mnd)$.
\end{theorem}

\begin{theorem}[same as Theorem \ref{thm:fast-VS}]
Using random projection for dimensionality reduction, the polynomial time algorithm for volume sampling mentioned in Theorem \ref{thm:polytimeVS} (i.e., Algorithm \ref{alg:outline} with Algorithm \ref{alg:omegaSub} as its subroutine), gives $(1+\eps)$-approximate volume sampling, using
\[
O\left(mn \log m \cdot \frac{k^{2}}{\eps^{2}} + m \log^{\omega} m \cdot \frac{k^{2\omega + 1}}{\eps^{2\omega}} \log(k \eps^{-1} \log m) \right).
\]
arithmetic operations.
\end{theorem}
\begin{proof}
Using Theorem \ref{thm:random-projection} and doing volume sampling of $k$-subsets of rows from $\tilde{A}$ gives $(1+\eps)$-approximation to the volume sampling of $k$-subsets of rows from $A$. This can be done in two steps: first, we compute $\tilde{A}$ using matrix multiplication in time $O(mnd)$ and second, we do volume sampling on $\tilde{A}$ using the algorithm from Subsection \ref{subsec:no-svd}. Overall, it takes time $O(mnd + kmd^{\omega} \log d)$, which is equal to
\[
O\left(mn \log m \cdot \frac{k^{2}}{\eps^{2}} + m \log^{\omega} m \cdot \frac{k^{2\omega + 1}}{\eps^{2\omega}}\right).
\]
Moreover, this can be implemented using only one pass over the matrix $A$ with extra space $m \log m \cdot k^{2}/\eps^{2}$.
\end{proof}

\section{Derandomized row/column-subset selection} \label{sec:derand}
Our derandomized row-subset selection algorithm is based on a derandomization of the volume sampling algorithm in Section \ref{sec:algo-outline}, using the method of conditional expectations. Again, it may be easier to consider volume sampling extended to random $k$-tuples $(X_{1}, \dotsc, X_{k})$ where
\[
\prob{X_{1} = i_{1}, \dotsc, X_{k} = i_{k}} = \begin{cases} \dfrac{\deter{A_{\{i_{1}, \dotsc, i_{k}\}} A_{\{i_{1}, \dotsc, i_{k}\}}^{T}}}{k! \sum_{S \subseteq [m] \suchthat \size{S}=k} \deter{A_{S} A_{S}^{T}}} & \text{if $i_{1}, \dotsc, i_{k}$ are distinct} \\ \quad & \quad \\ 0 & \text{otherwise} \end{cases}
\]
From Theorem \ref{thm:vol-sampling} we know that
\[
\expec{\frob{A - \pi_{\{X_{1}, \dotsc, X_{k}\}}(A)}^{2}} \leq (k+1) \frob{A - A_{k}}^{2},
\]
where the expectation is over $(X_{1}, \dotsc, X_{k})$.

Let us consider $i_{1}, \dotsc, i_{t-1}$ for which $\prob{X_{1} = i_{1}, \dotsc, X_{t-1} = i_{t-1}} > 0$. Let $S = \{i_{1}, \dotsc, i_{t-1}\}$ and look at the conditional expectation. The following lemma shows that these conditional expectations have an easy interpretation in terms of the coefficients of certain characteristic polynomials, and hence can be computed efficiently.

\begin{lemma} \label{lemma:conditional}
Let $(i_{1}, \dotsc, i_{t-1}) \in [m]^{t-1}$ be such that $\prob{X_{1} = i_{1}, \dotsc, X_{t-1} = i_{t-1}} > 0$ for a random $k$-tuple $(X_{1}, X_{2}, \dotsc, X_{k})$ from extended volume sampling. Let $S = \{i_{1}, \dotsc, i_{t-1}\}$ and $B = A - \pi_{S}(A)$. Then
\[
\expec{\frob{A - \pi_{\{X_{1}, \dotsc, X_{k}\}}(A)}^{2} \given X_{1} = i_{1}, \dotsc, X_{t-1} = i_{t-1}} = \frac{(k-t+2) c_{m-k+t-2} (BB^{T})}{c_{m-k+t-1} (BB^{T})}.
\]
\end{lemma}
\begin{proof}
\begin{align*}
& \expec{\frob{A - \pi_{\{X_{1}, \dotsc, X_{k}\}}(A)}^{2} \given X_{1} = i_{1}, \dotsc, X_{t-1} = i_{t-1}} \\
& = \sum_{(i_{t}, \dotsc, i_{k}) \in [m]^{k-t+1}} \frob{A - \pi_{\{i_{1}, \dotsc, i_{k}\}}(A)}^{2} \prob{X_{1} = i_{1}, \dots, X_{k} = i_{k} \given X_{1} = i_{1}, \dotsc, X_{t-1} = i_{t-1}} \\
& = \sum_{(i_{t}, \dotsc, i_{k}) \in [m]^{k-t+1}} \frob{A - \pi_{\{i_{1}, \dotsc, i_{k}\}}(A)}^{2} \frac{\prob{X_{1} = i_{1}, \dotsc, X_{k} = i_{k}}}{\prob{X_{1} = i_{1}, \dotsc, X_{t-1} = i_{t-1}}} \\
& = \sum_{(i_{t}, \dotsc, i_{k}) \in [m]^{k-t+1}} \frac{\sum_{l=1}^{m} \norm{d_{l}}^{2} \deter{A_{\{i_{1}, \dotsc, i_{k}\}} A_{\{i_{1}, \dotsc, i_{k}\}}^{T}}}{\sum_{(j_{t}, \dotsc, j_{k}) \in [m]^{k-t+1}} \deter{A_{\{i_{1}, \dotsc, i_{t-1}, j_{t}, \dots, j_{k}\}} A_{\{i_{1}, \dotsc, i_{t-1}, j_{t}, \dots, j_{k}\}}^{T}}} \\
& \qquad \text{where $D = A - \pi_{\{i_{1}, \dotsc, i_{k}\}}(A)$} \\
& = \frac{\sum_{(i_{t}, \dotsc, i_{k}) \in [m]^{k-t+1}} \sum_{l \notin \{i_{1}, \dotsc, i_{k}\}} \deter{A_{\{l, i_{1}, \dotsc, i_{k}\}} A_{\{l, i_{1}, \dotsc, i_{k}\}}^{T}}}{\sum_{(j_{t}, \dotsc, j_{k}) \in [m]^{k-t+1}} \deter{A_{\{i_{1}, \dotsc, i_{t-1}, j_{t}, \dots, j_{k}\}} A_{\{i_{1}, \dotsc, i_{t-1}, j_{t}, \dots, j_{k}\}}^{T}}} \\
& = \frac{(k-t+1)! \sum_{T \subseteq [m] \suchthat \size{S \cup T}=k, \size{T}=k-t+1} \sum_{l \notin S \cup T} \deter{A_{\{l\} \cup S \cup T} A_{\{l\} \cup S \cup T}^{T}}}{(k-t+1)! \sum_{T \subseteq [m] \suchthat \size{S \cup T}=k, \size{T}=k-t+1} \deter{A_{S \cup T} A_{S \cup T}^{T}}} \\
& = \frac{\sum_{T \subseteq [m] \suchthat \size{S \cup T}=k, \size{T}=k-t+1} \sum_{l \notin S \cup T} \deter{A_{S} A_{S}^{T}} \deter{B_{\{l\} \cup T} B_{\{l\} \cup T}^{T}}}{\sum_{T \subseteq [m] \suchthat \size{S \cup T}=k, \size{T}=k-t+1} \deter{A_{S} A_{S}^{T}} \deter{B_{T} B_{T}^{T}}} \quad \text{by Lemma \ref{lemma:det-division}} \\
& = \frac{\sum_{T \subseteq [m] \suchthat \size{S \cup T}=k, \size{T}=k-t+1} \sum_{l \notin S \cup T} \deter{B_{\{l\} \cup T} B_{\{l\} \cup T}^{T}}}{\sum_{T \subseteq [m] \suchthat \size{S \cup T}=k, \size{T}=k-t+1} \deter{B_{T} B_{T}^{T}}} \\
& = \frac{(k-t+2) \sum_{T \subseteq [m] \suchthat \size{S \cup T}=k+1, \size{T}=k-t+2} \deter{B_{T} B_{T}^{T}}}{\sum_{T \subseteq [m] \suchthat \size{S \cup T}=k, \size{T}=k-t+1} \deter{B_{T} B_{T}^{T}}} \\
& = \frac{(k-t+2) \sum_{T \subseteq [m] \suchthat \size{T}=k-t+2} \deter{B_{T} B_{T}^{T}}}{\sum_{T \subseteq [m] \suchthat \size{T}=k-t+1} \deter{B_{T} B_{T}^{T}}} \quad \substack{\text{the extra terms in the numerator and} \\ \text{the denominator are zero}} \\
& = \frac{(k-t+2) \abs{c_{m-k+t-2} (BB^{T})}}{\abs{c_{m-k+t-1} (BB^{T})}} \quad \text{by Lemma \ref{lemma:char-poly}}.
\end{align*}
\end{proof}

Knowing the above lemma, it is easy to derandomize our algorithm outlined for volume sampling. In each step, we just compute the new conditional expectations for each additional $i$, and finally pick the $i$ that minimizes the conditional expectation.

\begin{framed}
\begin{alg} \label{alg:derand}
{\bf Derandomized row/column-subset selection}
\end{alg}
\noindent Input: a matrix $A \in \reals^{m \times n}$ and $1 \leq k \leq \rank(A)$.

\noindent Output: a subset $S$ of $k$ rows of $A$ with the guarantee
\[
\frob{A - \pi_{S}(A)}^{2} \leq (k+1) \frob{A - A_{k}}^{2}.
\]

\begin{enumerate}
\item Initialize $S \leftarrow \emptyset$ and $B \leftarrow A$. For $t=1$ to $k$ do:
\begin{enumerate}
\item For $i=1$ to $m$ do: compute $c_{n-k+t-1} (C_{i}^{T} C_{i})$ and $c_{n-k+t} (C_{i}^{T} C_{i})$, where $C_{i} = B - \pi_{\{i\}}(B)$ is the matrix obtained by projecting each row of $B$ orthogonal to $b_{i}$.
\item Pick $i$ that minimizes $\abs{c_{n-k+t-1} (C_{i}^{T} C_{i})} / \abs{c_{n-k+t} (C_{i}^{T} C_{i})}$. Let $S \leftarrow S \cup \{i\}$ and $B \leftarrow C_{i}$.
\end{enumerate}
\item Output $S$.
\end{enumerate}
\end{framed}

\begin{theorem}[same as Theorem \ref{thm:det-sub-select}]
Deterministic Algorithm \ref{alg:derand}, when given a matrix $A \in \reals^{m \times n}$ and an integer $1 \leq k \leq \rank(A)$, outputs a $k$-subset $S$ of the rows of $A$, using $O(kmn^{\omega} \log n)$ arithmetic operations, such that
\begin{align*}
\frob{A - \pi_{S}(A)} & \leq \sqrt{k+1} \frob{A - A_{k}} \\
\norm{A - \pi_{S}(A)}_{2} & \leq \sqrt{(k+1)(n-k)} \norm{A - A_{k}}_{2}.
\end{align*}
\end{theorem}
\begin{proof}
By applying Lemma \ref{lemma:conditional} to $S \cup \{i\}$ instead of $S$, as $C_{i} = B - \pi_{\{i\}}(B) = A - \pi_{S \cup \{i\}}(A)$, we see that the step $t$ of our algorithm picks $i$ that minimizes
\begin{align*}
\expec{\frob{A - \pi_{\{X_{1}, \dotsc, X_{k}\}}(A)}^{2} \given X_{1} = i_{1}, \dots, X_{t-1} = i_{t-1}, X_{t} = i} & = \frac{(k-t+1) \abs{c_{n-k+t-1} (C_{i}^{T} C_{i})}}{\abs{c_{n-k+t} (C_{i}^{T} C_{i})}} \\
 & = \frac{(k-t+1) \abs{c_{m-k+t-1} (C_{i} C_{i}^{T})}}{\abs{c_{m-k+t} (C_{i} C_{i}^{T})}}.
\end{align*}
The correctness of our algorithm follows immediately from observing that in each step $t$,
\begin{align*}
& \expec{\frob{A - \pi_{\{X_{1}, \dotsc, X_{k}\}}(A)}^{2} \given X_{1} = i_{1}, \dots, X_{t-1} = i_{t-1}} \\
& = \sum_{i=1}^{m} \prob{X_{t} = i \given X_{1} = i_{1}, \dotsc, X_{t-1} = i_{t-1}} \expec{\frob{A - \pi_{\{X_{1}, \dotsc, X_{k}\}}(A)}^{2} \given X_{1} = i_{1}, \dots, X_{t-1} = i_{t-1}, X_{t} = i}
\end{align*}
and that we started with
\[
\expec{\frob{A - \pi_{\{X_{1}, \dotsc, X_{k}\}}(A)}^{2}} \leq (k+1) \frob{A - A_{k}}^{2}.
\]

The guarantee for spectral norm follows immediately from our guarantee in the Frobenius norm, just using properties of norms and the fact that $\rank(A - A_{k}) \leq n-k$:
\[
\snorms{A - \pi_{S}(A)} \leq \norms{A - \pi_{S}(A)}_F \leq (k+1) \norms{A - A_{k}}_F \leq (k+1)(n-k)\snorms{A - A_{k}}.
\]

Moreover, this algorithm runs in time $O(kmn^{\omega} \log n)$ if we use the subroutine in Subsection \ref{subsec:no-svd} to compute the characteristic polynomial of $C_{i} C_{i}^{T}$ using that of $C_{i}^{T} C_{i}$.
\end{proof}

\section{Lower bound for rank-$1$ spectral approximation using one row}
Here we show a lower bound for row/column-subset selection. We prove that there is a matrix $A \in \reals^{n \times (n+1)}$ such that using the span of any single row of it, we can get only $\Omega(\sqrt{n})$-approximation in the spectral norm for the nearest rank-$1$ matrix to $A$. This can be generalized to a similar $\Omega(\sqrt{n})$ lower bound for general $k$ by using a matrix with $k$ block-diagonal copies of $A$.

\begin{theorem}[same as Theorem \ref{thm:lower-bound}]
There exists a matrix $A \in \reals^{n \times (n+1)}$ such that
\[
\norm{A - \pi_{\{i\}}(A)}_{2} = \Omega(\sqrt{n}) \norm{A - A_{1}}_{2}, \quad \text{for all $1 \leq i \leq n$},
\]
where $\pi_{\{i\}}(A) \in \reals^{n \times (n+1)}$ is the matrix obtained by projecting each row of $A$ onto the span of its $i$-th row $a_{i}$.
\end{theorem}
\begin{proof}
Consider $A \in \reals^{n \times (n+1)}$ with entries as follows:
\[
\left(\begin{array}{ccccc}
1 & \eps & 0 & \dotsc & 0 \\
1 & 0 & \eps & \dotsc & 0 \\
1 & 0 & 0 & \dotsc & 0 \\
1 & 0 & \dotsc & 0 & \eps
\end{array}\right), \qquad 0 < \eps < 1.
\]
Let $B$ be the best rank-$1$ approximation to $A$ whose rows lie in the span of $(1, \eps, 0, \dotsc, 0)$ (or for that matter, any fixed row of $A$). Then, we want to show that
\[
\norm{A - B}_{2} \geq \frac{\sqrt{n}}{2} \norm{A - A_{1}}_{2} = \frac{\sqrt{n}}{2} \sigma_{2}(A).
\]
We first compute the singular values of $A$, i.e., the positive square roots of the eigenvalue of $AA^{T} \in \reals^{n \times n}$.
\[
AA^{T} = \left(\begin{array}{ccccc}
1+\eps^{2} & 1 & 1 & \dotsc & 1 \\
1 & 1+\eps^{2} & 1 & \dotsc & 1 \\
1 & 1 & \dotsc & 1 & 1 \\
1 & \dotsc & 1 & 1+\eps^{2} & 1 \\
1 & \dotsc & 1 & 1 & 1+\eps^{2}
\end{array}\right).
\]
$(1, 1, \dotsc, 1)$ is an eigenvector of $AA^{T}$ with eigenvalue $n+\eps^{2}$. Thus, $\sigma_{1}(A) = \sqrt{n+\eps^{2}}$. Observe that, by symmetry, all other singular values of $A$ must be equal, i.e., $\sigma_{2}(A) = \sigma_{3}(A) = \dotsc = \sigma_{n}(A)$. However,
\[
\frob{A}^{2} = \sum_{ij} A_{ij}^{2} = n+n\eps^{2} = \sum_{i=1}^{n} \sigma_{i}(A)^{2} = \sigma_{1}(A)^{2} + (n-1)\sigma_{2}(A)^{2} = n+\eps^{2} + (n-1)\sigma_{2}(A)^{2}.
\]
Therefore, $\norm{A-A_{1}}_{2} = \sigma_{2}(A) = \eps$.

Now denote the $i$-th row of $A$ by $a_{i}$. By definition, the $i$-th row of $B$ is the projection of $a_{i}$ onto $\linspan{a_{1}}$. We are interested in the singular values of $A-B$. For $i \geq 2$:
\begin{align*}
a_{i} - b_{i} & = a_{i} - \frac{\inner{a_{i}}{a_{1}}}{\norm{a_{1}}^{2}} a_{1} \\
& = \left(\frac{\eps^{2}}{1+\eps^{2}}, \frac{-\eps}{1+\eps^{2}}, 0, \underset{\overbrace{\text{$(i+1)$-th coord.}}}{\eps}, 0, \dotsc, 0\right).
\end{align*}
Thus, $(A-B)(A-B)^{T} \in \reals^{n \times n}$ can be written as
\[
(A-B)(A-B)^{T} = \left(\begin{array}{ccccc}
0 & 0 & \dotsc & 0 & 0 \\
0 & \frac{\eps^{2}(2+\eps^{2})}{1+\eps^{2}} & \frac{\eps^{2}}{1+\eps^{2}} & \dotsc & \frac{\eps^{2}}{1+\eps^{2}} \\
\dotsc & \frac{\eps^{2}}{1+\eps^{2}} & \frac{\eps^{2}(2+\eps^{2})}{1+\eps^{2}} & \dotsc & \frac{\eps^{2}}{1+\eps^{2}} \\
0 & \frac{\eps^{2}}{1+\eps^{2}} & \dotsc & \frac{\eps^{2}(2+\eps^{2})}{1+\eps^{2}} & \frac{\eps^{2}}{1+\eps^{2}} \\
0 & \frac{\eps^{2}}{1+\eps^{2}} & \dotsc & \frac{\eps^{2}}{1+\eps^{2}} & \frac{\eps^{2}(2+\eps^{2})}{1+\eps^{2}}
\end{array}\right).
\]
Again, $(0, 1, 1, \dotsc, 1)$ is the top eigenvector of $(A-B)(A-B)^{T}$ and using this we get,
\[
\norm{A-B}_{2}^{2} = \sigma_{1}(A-B)^{2} = \frac{\eps^{2}(2+\eps^{2})}{1+\eps^{2}} + (n-2) \frac{\eps^{2}}{1+\eps^{2}}.
\]
Therefore,
\[
\norm{A-B}_{2} = \frac{\eps}{\sqrt{1+\eps^{2}}} \sqrt{n+\eps^{2}} \geq \frac{\sqrt{n}}{2} \norm{A-A_{1}}_{2}.
\]
\end{proof}

\section{Discussion}


We analyzed efficient algorithms for volume sampling that can be used for row/column subset selection. Here are some ideas for future investigation suggested by this work:
\begin{itemize}
\item It would be interesting to explore how these algorithmic ideas are related to determinantal sampling \cite{lyons2003determinantal,hough2006determinantal} and, in particular, the generation of random spanning trees.

\item Find practical counterparts of the algorithms discussed here. In particular, we do not analyze the numerical stability of our algorithms.

\item Is there an efficient algorithm for volume sampling based on random walks? This question is inspired by MCMC as well as random walk algorithms for the generation of random spanning trees.
\end{itemize}

\bibliographystyle{abbrv}    

\bibliography{subset-ref}

\begin{thebibliography}{10}

\bibitem{BDM}
C.~Boutsidis, P.~Drineas, and M.~Mahoney.
\newblock An improved approximation algorithm for the column subset selection
  problem.
\newblock In {\em ACM-SIAM Symposium on Discrete Algorithms (SODA)}, 2009.

\bibitem{ACT}
P.~B{\"u}rgisser, M.~Clausen, and M.~A. Shokrollahi.
\newblock {\em Algebraic complexity theory}, volume 315 of {\em Grundlehren der
  Mathematischen Wissenschaften [Fundamental Principles of Mathematical
  Sciences]}.
\newblock Springer-Verlag, Berlin, 1997.
\newblock With the collaboration of Thomas Lickteig.

\bibitem{CM1}
A.~\c{C}ivril and M.~Magdon-Ismail.
\newblock On selecting the maximum volume submatrix of a matrix and related
  problems.
\newblock {\em Theoretical Computer Science}, 410 (47-49):4801--4811, 2009.

\bibitem{CM2}
A.~\c{C}ivril and M.~Magdon-Ismail.
\newblock Exponential inapproximability of selecting a maximum volume
  submatrix.
\newblock unpublished manuscript, 2010.

\bibitem{DRVW}
A.~Deshpande, L.~Rademacher, S.~Vempala, and G.~Wang.
\newblock Matrix approximation and projective clustering via volume sampling.
\newblock {\em Theory of Computing}, 2(1):225--247, 2006.

\bibitem{DV}
A.~Deshpande and S.~Vempala.
\newblock Adaptive sampling and fast low-rank matrix approximation.
\newblock In {\em International Workshop on Randomization and Computation
  (RANDOM)}, 2006.

\bibitem{DMM}
P.~Drineas, M.~Mahoney, and S.~Muthukrishnan.
\newblock Relative-error cur matrix decompositions.
\newblock {\em SIAM Journal of Matrix Analysis and Applications}, 30
  (2):844--881, 2008.

\bibitem{FKV}
A.~Frieze, R.~Kannan, and S.~Vempala.
\newblock Fast monte-carlo algorithms for finding low-rank approximations.
\newblock {\em Journal of the ACM (JACM)}, 51(6):1025--1041, 2004.

\bibitem{GVL}
G.~Golub and C.~van Loan.
\newblock {\em Matrix Computations}.
\newblock Johns Hopkins University Press, 1996.

\bibitem{GT}
S.~Goreinov and E.~Tyrtyshnikov.
\newblock The maximum-volume concept in approximation by low-rank matrices.
\newblock {\em Contemporary Mathematics}, 280:47--51, 2001.

\bibitem{GTZ}
S.~Goreinov, E.~Tyrtyshnikov, and N.~Zamarashkin.
\newblock Pseudo-skeleton approximations by matrices of maximal volume.
\newblock {\em Mathematicheskie Zametki}, 62:619--623, 1997.

\bibitem{GE}
M.~Gu and S.~Eisenstat.
\newblock Efficient algorithms for computing a strong rank-revealing {QR}
  factorization.
\newblock {\em SIAM Journal of Scientific Computing}, 17:848--869, 1996.

\bibitem{hough2006determinantal}
J.~Hough, M.~Krishnapur, Y.~Peres, and B.~Vir{\'a}g.
\newblock {Determinantal processes and independence}.
\newblock {\em Probability Surveys}, 3:206--229, 2006.

\bibitem{JV}
C.~Jeannerod and G.~Villard.
\newblock Essentially optimal computation of the inverses of generic polynomial
  matrices.
\newblock {\em Journal of Complexity}, 21 (1):72--86, 2005.

\bibitem{KV}
R.~Kannan and S.~Vempala.
\newblock Spectral algorithms.
\newblock {\em Foundations and Trends in Theoretical Computer Science},
  4:157--288, 2009.

\bibitem{lyons2003determinantal}
R.~Lyons.
\newblock {Determinantal probability measures}.
\newblock {\em Publications Math{\'e}matiques de l'IH{\'E}S}, 98(1):167--212,
  2003.

\bibitem{MZ}
A.~Magen and A.~Zouzias.
\newblock Near optimal dimensionality reductions that preserve volumes.
\newblock In {\em International Workshop on Randomization and Computation
  (RANDOM)}, 2008.

\bibitem{MD}
M.~Mahoney and P.~Drineas.
\newblock {CUR} matrix decompositions for improved data analysis.
\newblock {\em Proceedings of the National Academy of Sciences USA},
  106:697--702, 2009.

\bibitem{P}
C.-T. Pan.
\newblock On the existence and computation of rank-revealing {LU}
  factorizations.
\newblock {\em Linear Algebra and its Applications}, 316 (1-3):199--222, 2000.

\end{thebibliography}

\end{document}